\documentclass[11pt]{article}

\usepackage{graphicx}
\usepackage{latexsym}
\usepackage{fancyhdr}
\usepackage{subfigure}
\usepackage{mdwlist}

\makeatletter

\newenvironment{proof}{\paragraph{Proof:}}{\hspace*{\fill}\(\Box\)}
\newenvironment{argument}{\paragraph{Argument:}}{\hspace*{\fill}\(\Box\)}
\newtheorem{theorem}{Theorem}
\newtheorem{conjecture}{Conjecture}

\def\noflash#1{\setbox0=\hbox{#1}\hbox to 1\wd0{\hfill}}

\newcommand{\comment}[1]{}

\newcommand{\nocomment}[1]{}

\newcommand{\Iitemize}{\begin{itemize}
	{\setlength{\itemsep}{-6pt}}
       }

\newcommand{\ls}[1]
   {\dimen0=\fontdimen6\the\font 
    \lineskip=#1\dimen0
    \advance\lineskip.5\fontdimen5\the\font
    \advance\lineskip-\dimen0
    \lineskiplimit=.9\lineskip
    \baselineskip=\lineskip
    \advance\baselineskip\dimen0
    \normallineskip\lineskip
    \normallineskiplimit\lineskiplimit
    \normalbaselineskip\baselineskip
    \ignorespaces
   }

\def\ifundefined#1{\expandafter\ifx\csname#1\endcsname\relax}
\ifundefined{newblock}{
 }\fi

\newcommand{\eqref}[1]{Equation~\ref{#1}}

\newcommand{\mybox}[1]{\centerline{\framebox{\parbox[c]{\textwidth}{#1}}}}

\begin{document}

\title{Complexity of Multi-Value Byzantine Agreement \footnote{\normalsize This research is supported
in part by Army Research Office grant
W-911-NF-0710287. Any opinions, findings, and conclusions or recommendations expressed here are those of the authors and do not
necessarily reflect the views of the funding agencies or the U.S. government.}}

\date{\today}
\author{Guanfeng Liang and Nitin Vaidya\\ \normalsize Department of Electrical and Computer Engineering, and\\ \normalsize Coordinated Science Laboratory\\ \normalsize University of Illinois at Urbana-Champaign\\ \normalsize gliang2@illinois.edu, nhv@illinois.edu\\~\\Technical Report}

%


\maketitle


\thispagestyle{empty}

\newpage

\setcounter{page}{1}

\section{Introduction}\label{sec:intro}
In this paper, we consider the problem of maximizing the throughput of Byzantine agreement, given that the {\bf sum capacity} of all links in between nodes in the system is {\em finite}.
Byzantine agreement (BA) is a classical problem in distributed computing, with initial solutions presented
in the seminal work of Pease, Shostak and Lamport \cite{psl80,psl82}.
Many variations on the Byzantine {\em agreement} problem have been introduced in the past,
with some of the variations also called {\em consensus}. We will use the following definition
of Byzantine agreement (Byzantine general problem): Consider a network with
one node designated as the {\em sender} or {\em source} (S), and the other nodes
designated as the {\em peers}.
The goal of Byzantine agreement is for all the fault-free nodes to ``agree on''
the value being sent by the sender, despite the possibility that some of the
nodes may be faulty. In particular, the following conditions must be satisfied:

\begin{itemize}

\item
\textbf{Agreement:} All fault-free peers must agree on an identical value.

\item
\textbf{Validity:} If the sender is fault-free, then the agreed value
must be identical to the sender's value.
\item \textbf{Termination:} Agreement between fault-free peers is eventually achieved.
\end{itemize}

Our goal in this work is to design algorithms that can achieve optimal {\em throughput}\,
of agreement.
%
%
When defining throughput, the ``value'' referred in the above definition of
agreement is viewed as an infinite sequence of {\em information} bits. We assume that the information bits have already been compressed, such that for any subsequence of length $n>0$, the $2^n$ possible sequences are sent by the sender with equal probability. Thus, no set of information bits sent by the sender contains useful information about other bits. This assumption comes from the observation about ``typical sequences'' in Shannon's work \cite{shannon}. 


At each peer, we view
the agreed information as being represented in an array of infinite length.
Initially, none of the bits in this array at a peer have been agreed upon.
As time progresses, the array is filled in with agreed bits. In principle, the
array may not necessarily be filled sequentially. For instance, a peer may agree
on bit number 3 before it is able to agree on bit number 2. 
Once a peer agrees on any bit,
that agreed bit cannot be changed.

We assume that an agreement algorithm begins execution at time 0.
The system is assumed to be synchronous.
In a given execution of an agreement algorithm,
suppose that by time $t$ all the
fault-free peers have agreed upon bits 0 through $b(t)-1$, and at least one
fault-free peer has not yet agreed on bit number $b(t)$.
Then, the agreement {\em throughput}\, is defined as\footnote{\normalsize As a technicality,
we assume that the limit exists. The definitions here can be modified suitably
to take into account the possibility that the limit may not exist.}
$\lim_{t\rightarrow \infty} ~ \frac{b(t)}{t}$.
\\

\mybox{
{\bf Capacity} of agreement in a given network, for a given sender
and a given set of peers, is defined as the supremum of all achievable
agreement throughputs.
}

\section{Related Work}
\label{sec:related}

{\em Prior work on agreement or consensus:}
There has been significant research on agreement in presence of
{\em Byzantine} or {\em crash}\, failures,
theory (e.g., \cite{psl82,lynch:book96,attiya:book98}) and practice
(e.g., \cite{chandra07paxos,
 proactive_byzantine}) 
both. Perhaps
closest to our context is the work on {\em continuous consensus}
\cite{CC_Mizrahi_recovery,
Dwork_common_knowledge,CC_Mizrahi_common_knowledge} and
{\em multi-Paxos} \cite{Paxos_Lamport98thepart-time,chandra07paxos}
that considers agreement on a long sequence of
values. For our analysis of throughput as well, we will consider such a long
sequence of values. 
In fact, \cite{chandra07paxos} presents
measurements of {\em throughput} achieved in a system that uses
multi-Paxos. 
However, to the best of our knowledge, the past work on multi-Paxos and continuous consensus
has not addressed the problem of optimizing throughput
of agreement while considering the {\em capacities of the network links}.

%

In \cite{Bounds_BA_Dolev85}, it is proven that, without authentication, $O(n^2)$ bits are necessary to be communicated in total, in order to be able to agree on just 1 bit in the worst case when $t<n/3$. Algorithms have been derived to achieve this quadratic lower bound \cite{bit_optimal_89, opt_bit_Welch92}. Efforts have also been dedicated into tolerating more than failures and reducing the quadratic lower bound from \cite{Bounds_BA_Dolev85}, with the help of a public-key infrastructure and/or allowing a small probability of error \cite{Waidner96information-theoreticpseudosignatures, authenticated_BA_Dolev83}. However, all these algorithms have bit complexity $\Omega(n^2)$. In \cite{multi-valued_BA_PODC06}, a multi-valued BA algorithm is introduced with complexity $O(nl+n^3(n+\kappa))$, where $\kappa$ denotes a security parameter The quadratic lower bound is overcome by reducing a BA with a long message of $l$ bits to BA with much shorter messages, using a universal hash function and allowing a small probability of error. On contrary, our algorithm is able to achieve agreement with a lower complexity even in the worst case.

Past work has explored the use of error-correcting
codes for asynchronous consensus (e.g., \cite{Friedman_DISC02}).
Our algorithms also use error detecting codes, but somewhat differently.

{\em Prior work on multicast using network coding:}
While the early work on fault tolerance typically 
relied on replication \cite{cooper85}
or source coding \cite{Rabin89efficientdispersal} as mechanisms
for tolerating packet tampering, {\em network coding} has been recently
used with significant success as a mechanism for tolerating attacks or failures.
In traditional routing algorithms, a node serving as a router, simply
forwards packets on their way to a destination. With network coding,
a node may ``mix'' (or {\em code}) packets
from different
neighbors \cite{Li03linearnetwork}, and forward the coded packets.
This approach has been demonstrated
to improve throughput, being of particular benefit in {\em multicast}
scenarios \cite{Li03linearnetwork,Koetter01analgebraic,feedback-medard}. The problem of {\em multicast} 
is related to {\em agreement}. 
There has been much research on multicast with network coding in presence
of a Byzantine attacker (e.g., \cite{
Cai06networkerror,
Ho04byzantinemodification, 
Jaggi_infocom07,
Signature_Infocom08,
Ho2010ITA,
homomorphic}).
The significant difference between Byzantine agreement and multicasting
is that the multicast problem formulation assumes that the
source of the data is always fault-free.

In our previous technical reports \cite{techreport_CBA,techreport_CBA_final,techreport_Byzantine_complete}, we characterized the agreement capacity of general four-node networks with finite link capacity (some links may not exist). Byzantine agreement algorithm that achieves throughput arbitrarily close to the capacity are proposed for both complete and incomplete four-node networks.

\section{Models}\label{sec:model}
We assume a synchronous network fully connected network of $n$ nodes, the IDs (identifiers) are common knowledge. We assume that all communication channels/links are private and that whenever a node sends a message directly to another, the identity of the sender is known to the recipient, but we otherwise make no cryptographic assumptions. We assume a strong adversary. That is, the adversary has complete knowledge on the BA algorithm and the information being sent by every node. The adversary can take over nodes at any point during the algorithm up to the point of taking over up to a $t<n/3$ nodes, including the source. The compromised nodes can engage in any kind of deviations from the algorithm, including false messages and collusion, or crash failures, while the remaining nodes are good and follow the algorithm.

%
%
%
%

\section{Results}\label{sec:results}
All our results are about achieving agreement deterministically, 
which means that under our algorithm, it is impossible for the good nodes to decide on different values. 

\begin{itemize}
\item We show that there exists an algorithm which computes Byzantine
agreement on an $l$-bit message in a network with $n$ nodes and at most $t$ faulty nodes, and uses $\frac{n(n-1)}{n-t}l + l^{1/2}O(n^4)$ bits of communication.

If we let $l$ approach infinity and consider the average number of bits used for agreeing on 1 bit, then we have the following result on the sum capacity

\item For the network with $n$ nodes and at most $t$ faulty nodes to achieve agreement throughput of $R$ bits/unit time, it is sufficient to have sum capacity $> \frac{n(n-1)}{n-t}R$.
\end{itemize}

\section{Byzantine Agreement Algorithm to Tolerate up to $t$ Faults}\label{sec:algorithm}
The proposed Byzantine agreement algorithm progresses in generations. In each generation, $(n-t)c$ bits are being agreed upon.  Normally (when no failure is detected), the length of each generation is chosen such that slightly more than $n(n-1)c$ bits can be transmitted in the network. The algorithm starts assuming no node is faulty.

\subsection{Operations when no failure detected}\label{subsec:initial}
\paragraph{Round 1:} The source node 0 divides the $(n-t)c$ information bits into $n-t$ packets of size $c$ bits, each packet being a symbol from GF($2^c$). Then node 0 encodes the $n-t$ packets of data into $2(n-1)$ packets, each of which is obtained as a linear combination of the $n-t$ packets of data. Let us denote the $n-t$ data packets as the data vector
\begin{equation}
\tilde{x} = [x_1,x_2,\dots,x_{n-t}]
\end{equation}
and the $2(n-1)$ coded packets as
\begin{equation}
y_1,y_2,\dots,y_{2(n-1)}.
\end{equation}
For the correctness of the algorithm, these $2(n-1)$ symbols (or packets) need to be computed such that any subset of $n-t$ encoded symbols constitutes independent linear combinations of the $n-t$ data packets. As we know from the design of Reed-Solomon codes, if $c$ is chosen large enough, this linear independence requirement can be satisfied. The weights or coefficients used to compute the linear combinations is part of the algorithm specification, and is assumed to be correctly known to all nodes a priori. Due to the above independence property, any $n-t$ of the $2(n-1)$ symbols - if they are not tampered - can be used to (uniquely) solve for the $n-t$ data symbols.

In round 1, node 0 transmits 2 packets $y_i, y_{n-1+i}$ to node $i$. Other nodes do not transmit.

\paragraph{Round 2:} Each peer node $i$ sends packet $y_i$ to every other peer node.

\paragraph{Round 3:} Each fault-free peer finds the solution for each subset of $n-t$ packets from among the packets received from the other nodes in rounds 1 and 2. If the solutions to the various subsets are not unique, then the fault-free peer has detected faulty behavior by some node in the network. In round 3, each peer broadcasts to the remaining $n-1$ nodes a 1-bit notification in indicating whether it has detected a failure or not -- to agreement among all $n$ nodes on these 1-bit indicators is achieved by using an efficient traditional  Byzantine agreement algorithm (e.g. the algorithm \cite{opt_bit_Welch92, bit_optimal_89}). Since less than 1/3 of the nodes can be faulty, using this traditional algorithm, all fault-free nodes obtain identical 1-bit notifications from all the peers. If none of the notifications indicate a detected failure, then each fault-free peer agrees on the unique solution obtained above, and the execution of the current generation of the algorithm is completed. However, if failure detection is indicated by any peer, then an ``extended round 3'' is added to the execution, as elaborated soon.\\

\hrule
\begin{theorem}\label{thm:detect}
Misbehavior by a faulty node will either be detected by at least one fault-free
peer, or all the fault-free peers will reach agreement correctly.
\end{theorem}

\begin{proof}
Node 0 is said to misbehave only if it sends packets to the peers
that are inconsistent -- that is, all of them are not appropriate
linear combinations of an identical data vector $\tilde{x}$. 
A peer node $i$ is said to misbehave if it forwards tampered (or incorrect)
packets other than $y_i$ to the other peers.
\\

\noindent
{\em Source is faulty:} Since the source node 0 is faulty, there are at least $n-t$ fault-free peers. From the way the peers forward their packets, it is easy to see that among the packets the fault-free peers have received, at least $n-t$ are identical. Thus, the fault-free peers will not agree on different data symbols (since the agreed data must satisfy linear equations corresponding to all received
packets), even though node 0 and up to $t-1$ peers may be misbehaving.
\\

\noindent
{\em Source is fault-free:} When source is fault-free, there are at least $n-t-1$ fault-free peers. Similar to the previous case, among the packets a fault-free peer have received, at least $n-t-2$ are received from fault-free the other peers and must be  correct. In addition, every peer has also received two correct packets from source node 0. Together, every fault-free peers has receives at least $n-t$ correct packets, which has a unique solution $\tilde{x}$. Thus, the fault-free peers will not agree on any data vector other than $\tilde{x}$, even though up to $t$ peers may be misbehaving.
\end{proof}
\hrule

\paragraph{Extended Round 3:} When failure detection is indicated by any peer in round 3, an extended round is added subsequent to a failure detection.
We also refer to this extended round as the ``broadcast phase''. During the  broadcast phase, every node (including the source) broadcasts all the packets it has sent to other nodes, or received from other nodes, during rounds 1 and 2 -- as with the failure notifications in round 3, agreements on these broadcast packets are achieved using the traditional Byzantine agreement algorithm.

Using the broadcast information extended round 3, all fault-free nodes form identical {\em diagnosis graphs} after the broadcast phase. The diagnosis graph contains $n$ vertices 0, 1, ..., $n-1$, corresponding to the $n$ nodes in our network; there is an undirected edge between each pair of vertices, with
each edge being labeled as $g$ at time 0 (with $g$ denoting ``good'').
The labels may change to $f$ during extended round 3. Once a label changes to $f$ (denoting ``faulty''), it is never changed back to $g$. Without loss
of generality, consider the edge between vertices X and Y in the diagnosis
graph. The label for this edge may be set to $f$ in two ways:
\begin{itemize}
\item
After the broadcast phase,
all nodes obtain identical information about what every node
``claims'' to have sent and received during rounds 1 and 2. Then, for each packet
in rounds 1 and 2 (sent by any node), each fault-free peer will compare the
claims by nodes X and Y about packets sent and received on links XY and YX.
If the two claims mismatch, then the label for edge XY in the
diagnosis graph is set to $f$.
\item If node X is a peer and claims to have detected a misbehavior in round 3, but the packets it claims to have received in rounds 1 and 2 are inconsistent with
this claim, then edge XY in diagnosis graph is set to $f$ for all Y$\neq$X. In this case, all edges associated with X are set to $f$.
\end{itemize}

Similar to the argument in \cite{techreport_Byzantine_complete}, an edge will be marked $f$ only if at least one of the nodes associated with this edge is faulty. Moreover, the label of at least one edge will be changed from $g$ to $f$ after the broadcast phase.

\subsection{Operations after failure detected}\label{subsec:detected}
After a failure is detected, and the extended round 3 is finished, a new generation of $(n-t)c$ bits of new data begins. Let us call nodes $i$ and $j$ accuse(trust) each other if edge $ij$ is marked $f$($g$) in the diagnosis graph. Notice that if a node is accused by more than $t$ other nodes, this node must be faulty. 
If the source node 0 is accused by more than $t$ peers, node 0 must be faulty, and the fault-free peers can terminate the algorithm and all agree on some default value. Now consider the case when the source node is accused by no more than $t$ peer. If a peer is identified as faulty, it is {\em isolated} from the network by marking all edges adjacent to it as $f$. Then the algorithm operates as follows.

\paragraph{Round 1:} Without loss of generality, assume that node 0 is accused by $m\le t<n/3$ peers. If node 0 is indeed fault-free, it encodes the $n-t$ packets of data into $2(n-1-m)>n-t$ packets\footnote{Since $m\le t$, $2(n-1-m)\ge 2(n-1-t) = n-t + (n-t-2)$. From $t\ge 1$ and $n\ge 3t+1$, we have $n-t-2>0$, then $2(n-1-m)>n-t$.}, each of which is obtained as a linear combination of the $n-t$ packets of data. And similar to the case when no failure is detected, every subset of $n-t$ encoded packets consists linear independent combinations of the $n-t$ data packets. For convenience, we will index the two packets node 0 sends to node i as $y_i$ and $y_{n-1+i}$ as before. For a node $i$ that accuses node 0, $y_i$ and $y_{n-1+i}$ do not exist (or equal to NULL).

\paragraph{Round 2:} Every fault-free peer $i$ that is trusted by node 0 sends $y_i$ to every peer $j$ that it trusts. 

\paragraph{Round 3:} For every fault-free peer $i$ that is accused by node 0, if it receives at least $n-t$ packets from the peers it trusts in round 2, peer $i$ first checks these packets for consistency. If these packets are inconsistent, node $i$ detects an attack. Otherwise node $i$ generates one packet $z_i$ as a linear combination of the packets it receives from the trusted peers such that $z_i$ and any $n-t-1$ packets from the trusted peers consists linear independent combinations of the $n-t$ data packets. It is possible that peer $i$ receives fewer than $n-t$ packets in round 2. For example: $t$ peers including node $i$ are accused by node 0, and $t$ other peers are accused by node $i$, then there are only $n-2t-1<n-t$ peers not accused by either node 0 or $i$. In this case, some of the trusted peers of node $i$ sends the second packet ($y_{n-1+j}$) to node $i$, such that node i receives $n-t$ packets in total. Then node $i$ generates the packet $z_i$ as a linear combination of the $n-t$ packets it has received similar to the previous case. After $z_i$ is generated, peer $i$ sends it to all the peers it trusts.

For the correctness of the algorithm, any subset of $n-t$ of the union of the $y$ and $z$ packets must be linearly independent. As we know from the design of Reed-Solomon codes and linear network coding, if $c$ is chosen large enough, this linear independence requirement can be satisfied. In principle, since every peer $i$ accused by the source receives at least $n-t$ linear independent packets, it can first solve $\tilde{x}$ and then generate $z_i$ in the same way the $y$ packets are generated.

After $z_i$ is generated, every node $i$ sends $z_i$ to all the peers it trusts. Then every fault-free peer finds the solution for each subset of $n-t$ packets from among the packets received from the other nodes in rounds 1, 2 and 3, and detects an attack if no unique solution is found. Then a 1-bit notification is broadcast as before. Similar to Theorem \ref{thm:detect}, we have the following

\hrule
\begin{theorem}
After a failure or attack is detected, further misbehavior by a faulty node will either be detected by at least one fault-free peer, or all the fault-free peers will reach agreement correctly.
\end{theorem}
\begin{proof}
Here we also consider the cases when source node 0 is fault-free and faulty separately.
\\

\noindent
{\em Source is faulty:} Denote $P$ and $Q$ as the set of fault-free peers that are trusted by the source and accused by the source, respectively. Since source is faulty, there are at least $n-t$ fault-free peers, i.e. $|P|+|Q|\ge n-t$. As stated before, an edge is marked $f$ only if at least one node associated with this edge is faulty. This implies that an edge associated with two fault-free nodes is always marked $g$, and two fault-free nodes will never accuse each other. Thus, every peer $i\in P$ sends $y_i$ to all peers in $P\bigcup Q$ during round 2. Similarly, every peer $i\in Q$  sends $z_i$ to all peers in $P\bigcup Q$ during round 3. As a result, the fault-free peers $P\bigcup Q$ share at least $|P|+|Q|\ge n-t$ identical packets. Thus, the fault-free peers will not agree on different data symbols.
\\

\noindent
{\em Source is fault-free:} Remember that two fault-free nodes will never accuse each other. Thus all fault-free peers receives two correct packets from source node 0 in round 1 and they exchange 1 packet with each other in round 2. Similar to the proof of Theorem \ref{thm:detect}, the fault-free peers will not agree on any vector other than $\tilde{x}$.
\end{proof}
\hrule

\paragraph{Extended Round 3:} If a failure is detected by a node associated with no more than $t$ $f$-links (otherwise the node must be faulty and will be isolated as before), the extended round 3 is entered and carried out as before. By the end of every extended round 3, at least one more edge in the diagnosis graph will be marked $f$.

\subsection{Complexity of the Algorithm}\label{subsec:complexity}
\paragraph{When no failure detected:} In every generation, excluding the extended round 3:
\begin{itemize}
\item Round 1: source node 0 sends $2(n-1)$ packets to the peers, which is $2(n-1)c$ bits;
\item Round 2: every peer sends $(n-2)c$ bits, thus $(n-1)(n-2)c$ bits in total;
\item Round 3: each peer broadcasts 1 bit notification. Let us denote $B$ as the bit-complexity of achieving agreement on 1 bit. So in round 3, totally $(n-1)B$ bits are transmitted;
\end{itemize}
So if no failure is detected in round 3, totally
\begin{equation}
2(n-1)c + (n-1)(n-2)c + (n-1)B = n(n-1)c + (n-1)B
\end{equation}
bits are transmitted in the whole generation.

\paragraph{After failure detected:} 
\begin{itemize}
\item Source node 0: sends at most $2(n-1)$ packets in round 1, which is no more than the usage when no failure detected;
\item Among nodes trusted by node 0: they send 1 packet to each other, which is the same as no failure detected;
\item Nodes accused by node 0: each such node $i$ receives and sends 1 packet from every peer it trusts, which is no more than the usage when no failure detected; if node $i$ receives fewer than $n-t$ packets in round 2, it receives $n-t$ in rounds 2 and 3 together, which is less than the usage of $n-2$ when no failure is yet detected. Every such node $i$ sends at most $n-2$ packets to the peers it trusts in round 3, which is no more than the usage when no failure is yet detected.
\item Only nodes accused by no more than $t$ nodes need to achieve agreement on the 1-bit notifications.
\end{itemize}
Now it should not be hard to see that after failure is detected, the usage of capacity is at most as much as the case when no failure is yet detected, excluding extended round 3. So in the normal operation (rounds 1 to 3), at most $n(n-1)c + (n-1)B$ bits are transmitted in every generation.

\paragraph{Extended round 3:} Source broadcasts at most $2(n-1)$ packets, and every peer broadcasts at most $n$ packets, when a failure is detected. So in every extended round 3, no more than $2(n-1)+n(n-1) = (n+2)(n-1)$ packets are broadcast, which results in $(n+2)(n-1)cB$ bits being transmitted.

Now we look at the total number of bits being used to agree on $l$ bits of data. Notice that $(n-t)c$ bits are being agreed on in every generation, so there are $l/(n-t)c$ generations in total. Thus in rounds 1 to 3 in all generations, no more than totally $\frac{n(n-1)}{n-t}l + \frac{(n-1)B}{(n-t)}\frac{l}{c}$ bits are transmitted. There are at most $(t+1)t$ extended round 3 before all faulty nodes are identified. So the total usage of extended round 3 in all generations are at most $(t+1)t(n+2)(n-1)cB$ bits. Now we have a upper bound on the total number of bits being transmitted to achieve agreement on $l$ bits as:
\begin{eqnarray}
&&\frac{n(n-1)}{n-t}l + \frac{(n-1)B}{n-t}\frac{l}{c} + (t+1)t(n+2)(n-1)cB\\
&=& \frac{n(n-1)}{n-t}l + 2Bl^{1/2} \sqrt{\frac{(t+1)t(n+2)(n-1)^2}{n-t}}~~~\left(\mbox{by a suitable choice of } c=\sqrt{\frac{(n-1)l}{(n-t)(t+1)t(n+2)(n-1)}}\right)\\
&<& \frac{n(n-1)}{n-t}l + 2Bl^{1/2} \sqrt{\frac{(n+2)^2(n-1)^3}{6n}}~~~\left(\mbox{maximized when } t=(n-1)/3\right)\\
&=& \frac{n(n-1)}{n-t}l + Bl^{1/2} \Theta(n^2).
\end{eqnarray}
Notice that broadcast algorithm of complexity $\Theta(n^2)$ are known \cite{bit_optimal_89}, so we assume $B=\Theta(n^2)$. Then we have the complexity of our algorithm for all $t<n/3$ upper bounded by
\begin{equation}
\frac{n(n-1)}{n-t}l + l^{1/2} \Theta(n^4) < \frac{3}{2}nl + l^{1/2} \Theta(n^4). \label{eq:upper_bits}
\end{equation}

For a given network with size $n$, if we normalize the left hand side of Eq.\ref{eq:upper_bits} over $l$, we obtain a upper bound on the sum capacity of the system for achieving agreement throughput of R bits/unit time as
\begin{equation}
\frac{n(n-1)}{n-t}R + \frac{\Theta(n^4)}{l^{1/2}}R \rightarrow \frac{n(n-1)}{n-t}R,~~as~ l\rightarrow\infty \label{eq:capacity}
\end{equation}

Table \ref{tab:complexity} states the overall bit complexities for agreeing on an $l$-bit message, both for the most efficient protocols in the literature and for the
protocol presented in this paper. In particular, compared with the best known algorithm in \cite{multi-valued_BA_PODC06}, our algorithm has strictly lower complexity when $l\ge \omega n^6$ for some  constant $\omega>0$. Moreover, the low complexity in \cite{multi-valued_BA_PODC06} is achieved by allowing a positive probability of error (fault-free nodes may decide on different values), while our algorithm is guaranteed to achieve agreement deterministically such that all fault-free nodes always agree on the same (correct) value.

\begin{table}[t]
\centering
\begin{tabular}{c|c|c|c}
 {\bf Lit.} & {\bf Complexity} & {\bf Authentication required?} & {\bf Probabilistic/deterministic agreement?}\\
\hline
\cite{bit_optimal_89} & $\Theta(n^2l)$ & {\bf No} & {\bf Deterministic}\\
\hline
\cite{Waidner96information-theoreticpseudosignatures} & $\Omega(n^2l+ n^6\kappa)$ & Yes & Probabilistic\\
\hline
\cite{authenticated_BA_Dolev83} & $\Omega(n^2l+n^3\kappa)$ & Yes & Probabilistic\\
\hline
\cite{multi-valued_BA_PODC06} & $(3n-1)l + \Theta(n^3(n+\kappa))$ & {\bf No} & Probabilistic\\
\hline
Ours & $<\frac{3}{2}nl+l^{1/2}\Theta(n^4)$ & {\bf No} & {\bf Deterministic}
\end{tabular}
\caption{Complexity of existing BA algorithms. The complexity of algorithms from\cite{bit_optimal_89,Waidner96information-theoreticpseudosignatures} is cited from \cite{multi-valued_BA_PODC06}.}
\label{tab:complexity}
\end{table}

\section{Optimality of our Algorithm}\label{sec:optimality}
Authors of \cite{multi-valued_BA_PODC06} have shown that their algorithm is {\em order-optimal} since any algorithm that achieves Byzantine agreement requires the communication of $\Omega(nl)$ bits, and their algorithm requires $(3n-1)l + \Theta(n^3(n+\kappa))$ bits. This results in requiring at least $(3n-1)R$ bits/unit time of sum capacity in order to achieve agreement throughput of $R$. If we define the average cost per agreed bit $\alpha$ as the ratio of the sum capacity over the achievable throughput of agreement, then $\alpha_{\cite{multi-valued_BA_PODC06}} = 3n-1$. On the other hand, according to Eq.\ref{eq:capacity}, our algorithm requires a lower ratio $\alpha_{our}=\frac{n(n-1)}{n-t}\le 3n/2<3n-1$. In fact, we believe that our algorithm is not only order-optimal, but also optimal in the sense of minimizing $\alpha$.

\begin{conjecture}\label{conj:lower}
In order to achieve agreement on $l$ bits, at least $\frac{n(n-1)}{(n-t)}l$ bits need to be transmitted in the network. This means that to achieve agreement throughput of $R$ bits/unit time, the sum capacity of the system must be at least $\frac{n(n-1)}{(n-t)}R$ bits/unit time.
\end{conjecture}

\begin{argument}

This conjecture is argued based on another conjecture as follows:

\begin{conjecture}\label{conj:check}
Given $k$ nodes each of which is assigned an arbitrary initial value of $l$ bits, at least $(k-1)l$ bits are necessary such that if the $k$ initial values are not identical, as least one node will detect it.
\end{conjecture}
Conjecture \ref{conj:check} is easily shown to be true when $l=1$. Since we assume 1 bit as the smallest unit of data, if two nodes communicate with each other, at least 1 bit is transmitted on one of the two links connecting the two nodes. We call such a pair of nodes {\em connected}. It should not be hard to see that just to check if all $k$ initial bits are identical, the network must be connected. And it is well-known from graph theory, that to connect $k$ nodes, at least $k-1$ links are needed. Thus, at least $k-1$ bits are necessary just to check initial values of 1 bit among $k$ nodes.  Intuitively, every bit of the $l$-bit initial value is independent. So the checking result of a particular bit is independent of the results of the other bits. This implies that every bit needs to be checked individually. For each bit, we need to form a connected graph using at least $k-1$ bits. Thus, to check $l$ bits, $(k-1)l$ bits are needed in total. Although this intuition sounds valid, the formal proof is much more intriguing and we are still working on it. This is the reason why we state this as a conjecture rather than a theorem.

Based on Conjecture \ref{conj:check}, we argue that in a network with $n$ nodes and up to $t$ failures, when no failure is yet detected initially, at least $(n-t-1)l$ bits need to be allocated for links among every subset of $n-t$ nodes. To argue this, we need the following reduction:

Remember that our goal is to achieve agreement with up to $t$ faults (any subset of $\le t$ nodes). Consider one subset $F=\{f_0,\dots,f_{t-1}\}$ of $t$ nodes containing the source node 0 ($f_0$) that may be faulty. The other $n-t$ nodes are known to be fault-free and denote the set of these nodes as $G=\{g_1,\dots,g_{n-t}\}$. 

Given any algorithm that achieve agreement on $l$ bits, we construct the state machine illustrated in Fig.\ref{fig:state_machine}. In this state machine, for node $g_i$, $F_i=\{f_{i,0},\dots,f_{i,t-1}\}$ is the set of virtual nodes corresponding to $F$ and run the same {\em correct} code as nodes in $F$ should run, and $g_{i,j}$ is the virtual node corresponding to $g_j$ and runs the same code as node $g_j$. The good node $g_i$ sends identical messages to node $g_j$ and $g_{i,j}$. Each virtual source node $f_{i,0}$ is given an initial value $v_i$ of $l$ bits. In the real network, node $f_j$ behaves to node $g_i$ as node $f_{i,j}$ in the state machine. It should be easy to see that if $v_1=v_2=\cdots =v_{n-t}$ the nodes in $B$ are actually not misbehaving.

Let us assume that all the nodes in $G$ know that nodes in $F$ behave in the way described above. Observe that the behavior $f_i$ is fully determined by $v_i$ and the messages node $g_i$ sends to $f_i$. So if $g_i$ knows the value of $v_i$, it can emulate the behavior of $f_i$. Then we can reduce the agreement problem of $n$ nodes to the following problem of $n-t$ fault-free nodes of $G$, in which each node $g_i$ is given an initial value $v_i$. Since algorithm can achieve agreement in the original network, agreement must be able to achieved by having node $g_i$ emulating the behavior of $f_i$. 

The above agreement algorithm for the $n-t$ fault-free nodes can be used to check if the $n-t$ initial values are identical: Denote $v$ as the value agreed upon by the agreement algorithm. If $v_1=v_2=\cdots =v_{n-t}$, it is as if no node is misbehaving in the original network and $v$ must equal to $v_i$ for all $i$. On the other hand, $v\neq v_i$ for some $i$, then the $n-t$ initial values must not be all identical, and node $g_i$ will detect it. By Conjecture \ref{conj:check}, it is impossible to check the if all initial values are identical with fewer than $(n-t-1)l$ bits being transmitted. This means that the BA algorithm must allocate at least $(n-t-1)l$ bits on links between nodes in $G$. 

\begin{figure}[t]
\centering
\includegraphics{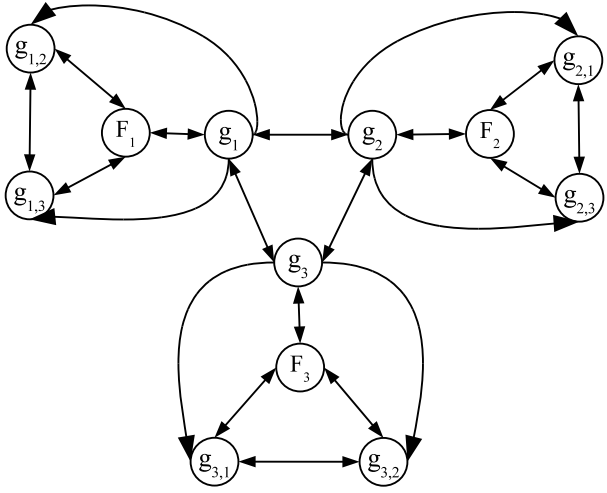}
\caption{State machine for $|G|=3$.}
\label{fig:state_machine}
\end{figure}

Since the location of the faulty peers are unknown a priori, so at least $(n-t-1)l$ bits must be assigned for every subset of $n-t$ peers. For the subsets containing the source and $n-t-1$ peers, a similar reduction is applied and leads to the same result.

Now we have shown that $(n-t-1)l$ bits are necessary for every subset of $n-t$ nodes in the network. If we sum over all ${n\choose n-t}$ subsets of $n-t$ nodes, the summation is ${n\choose n-t}(n-t-1)l$ bits, while each links is counted ${n-2\choose n-t-2}$ times.  Then we can computer the following lower bound on the total bit complexity for any BA algorithm on $l$ bits
\begin{eqnarray}
\frac{{n\choose n-t}(n-t-1)l}{{n-2\choose n-t-2}} &=& \frac{\frac{n!}{(n-t)!t!}(n-t-1)l}{\frac{(n-2)!}{t!(n-t-2)!}}=\frac{n(n-1)}{n-t}l
\end{eqnarray}
This completes the argument of Conjecture \ref{conj:lower}.
\end{argument}

If Conjecture \ref{conj:check} is proven to be true, Conjecture \ref{conj:lower} is true too. Then for a network of size $n$ with up to $t$ faults and total capacity $C$ bits/unit time, by Conjecture \ref{conj:lower}, the agreement capacity of this network cannot exceed $\frac{n-t}{n(n-1)}C$ bits/unit time. On the other hand, our algorithm can achieve agreement throughput arbitrarily close to this upper bound. Thus the agreement capacity of such a network is $\frac{n-t}{n(n-1)}C$ bits/unit time.

\section{Conclusion}
In this work, we have proposed a highly efficient Byzantine agreement algorithm on values of length $l>1$ bits. This algorithm uses error detecting network codes to ensure that fault-free nodes will never disagree, and routing scheme that is adaptive to the result of error detection. Our algorithm has a bit complexity of $\frac{n(n-1)}{n-t}l$, which leads to a linear cost ($O(n)$) per bit agreed upon, and overcomes the quadratic lower bound ($\Omega(n^2)$) in \cite{Bounds_BA_Dolev85}. Such linear per bit complexity has only been achieved in the literature by allowing a positive probability of error. Our algorithm achieves the linear per bit complexity while guaranteeing agreement is achieved correctly even in the worst case. We also conjecture that our algorithm can be used to achieve agreement throughput arbitrarily close to the agreement capacity of a network, when the sum capacity is given.

\bibliographystyle{abbrv}
\bibliography{PaperList}

\end{document}